\documentclass[a4paper,envcountsame]{llncs}

\usepackage{pst-plot,pst-node}
\usepackage{latexsym}
\usepackage{amssymb}
\usepackage{epsfig}
\usepackage{amsfonts}
\usepackage{color}
\usepackage{amsmath}
\usepackage{fleqn}
\usepackage[T1]{fontenc}
\usepackage{ifthen}

\newboolean{commentson} 
\setboolean{commentson}{true}

\newcommand{\comment}[1]
{\ifthenelse{\boolean{commentson}}
   {{\par\noindent\mbox{}{\small[ *** #1 ]\par}\noindent\par}}{}}

\begin{document}

\newcommand{\bull}{\rule{.85ex}{1ex} \par \bigskip}
\newenvironment{sketch}{\noindent {\bf Proof (sketch):\ }}{\hfill \bull}
\newtheorem{exmp}[theorem]{Example}
\newtheorem{notation}[theorem]{Notation}
\newtheorem{observation}[theorem]{Observation}

\newcommand{\QCSP}[1]{\mbox{\rm QCSP$(#1)$}}
\newcommand{\CSP}[1]{\mbox{\rm CSP$(#1)$}}
\newcommand{\MCSP}[1]{\mbox{{\sc Max CSP}$(#1)$}}
\newcommand{\wMCSP}[1]{\mbox{\rm weighted Max CSP$(#1)$}}
\newcommand{\cMCSP}[1]{\mbox{\rm cw-Max CSP$(#1)$}}
\newcommand{\tMCSP}[1]{\mbox{\rm tw-Max CSP$(#1)$}}
\renewcommand{\P}{\mbox{\bf P}}
\newcommand{\G}[1]{\mbox{\rm I$(#1)$}}
\newcommand{\NE}[1]{\mbox{$\neq_{#1}$}}

\newcommand{\MCol}[1]{\mbox{\sc Max $#1$-Col}}

\newcommand{\NP}{\mbox{\bf NP}}
\newcommand{\NL}{\mbox{\bf NL}}
\newcommand{\PO}{\mbox{\bf PO}}
\newcommand{\NPO}{\mbox{\bf NPO}}
\newcommand{\APX}{\mbox{\bf APX}}
\newcommand{\Aut}{\mbox{\rm Aut}}
\newcommand{\bound}{\mbox{\rm -$B$}}

\newcommand{\GIF}[3]{\ensuremath{h_\{{#2},{#3}\}^{#1}}}

\newcommand{\Spmod}{\mbox{\rm Spmod}}
\newcommand{\Sbmod}{\mbox{\rm Sbmod}}

\newcommand{\Inv}[1]{\mbox{\rm Inv($#1$)}}
\newcommand{\Pol}[1]{\mbox{\rm Pol($#1$)}}
\newcommand{\sPol}[1]{\mbox{\rm s-Pol($#1$)}}

\newcommand{\un}{\underline}
\newcommand{\ov}{\overline}
\def\ar{\hbox{ar}}
\def\vect#1#2{#1 _1\zdots #1 _{#2}}
\def\zd{,\ldots,}
\let\sse=\subseteq
\let\la=\langle
\def\lla{\langle\langle}
\let\ra=\rangle
\def\rra{\rangle\rangle}
\let\vr=\varrho
\def\vct#1#2{#1 _1\zd #1 _{#2}}
\newcommand{\va}{{\bf a}}
\newcommand{\vb}{{\bf b}}
\newcommand{\vc}{{\bf c}}
\newcommand{\bx}{{\bf x}}
\newcommand{\by}{{\bf y}}
\def\Z{{\bur Z^+}}
\def\R{{\bur R}}
\def\D{{\cal D}}
\def\F{{\cal F}}
\def\I{{\cal I}}
\def\C{{\cal C}}
\def\U{{\cal U}}
\def\K{{\cal K}}
\def\Lat{{\cal L}}

\def\2mat#1#2#3#4#5#6#7#8{
\begin{array}{c|cc}
$~$ & #3 & #4\\
\hline
#1 & #5& #6\\
#2 & #7 & #8 \end{array}}

\renewcommand{\phi}{\varphi}
\renewcommand{\epsilon}{\varepsilon}

\def\tup#1{\mathchoice{\mbox{\boldmath$\displaystyle#1$}}
{\mbox{\boldmath$\textstyle#1$}}
{\mbox{\boldmath$\scriptstyle#1$}}
{\mbox{\boldmath$\scriptscriptstyle#1$}}}
\newcommand{\draft}{\begin{center}\huge Draft!!! \end{center}}
\newcommand{\void}{\makebox[0mm]{}}     


\renewcommand{\text}[1]{\mbox{\rm \,#1\,}}        


\renewcommand{\emptyset}{\varnothing}  
\newcommand{\union}{\cup}               
\newcommand{\intersect}{\cap}           
\newcommand{\setdiff}{-}                
\newcommand{\compl}[1]{\overline{#1}}   
\newcommand{\card}[1]{{|#1|}}           
\newcommand{\set}[1]{\{{#1}\}} 
\newcommand{\st}{\ |\ }                 
\newcommand{\suchthat}{\st}             
\newcommand{\cprod}{\times}             
\newcommand{\powerset}[1]{{\bf 2}^{#1}} 

\newcommand{\tuple}[1]{\langle{#1}\rangle}  
\newcommand{\seq}[1]{\langle #1 \rangle}
\newcommand{\emptyseq}{\seq{}}
\newcommand{\floor}[1]{\left\lfloor{#1}\right\rfloor}
\newcommand{\ceiling}[1]{\left\lceil{#1}\right\rceil}

\newcommand{\map}{\rightarrow}
\newcommand{\fncomp}{\!\circ\!}         

\newcommand{\transclos}[1]{#1^+}
\newcommand{\reduction}[1]{#1^-}        


\newcommand{\perfimp}{\stackrel{p}{\Longrightarrow}}

\newcommand{\ie}{{\em ie.}}                
\newcommand{\eg}{{\em eg.}}
\newcommand{\paper}{paper}                

\newcommand{\emdef}{\em}                   
\newcommand{\rinterpretation}{${\mathbb R}$-interpretation}
\newcommand{\rmodel}{${\mathbb R}$-model}
\newcommand{\transp}{^{\rm T}}

\newcommand{\unprint}[1]{}
\newcommand{\blankline}{$\:$}

\newcommand{\prob}[1]{{\sc #1}}

\newcommand{\Solv}{{\it TSolve}}
\newcommand{\Neg}{{\it Neg}}
\newcommand{\logname}{XX}

\newcommand{\props}{{\it props}}
\newcommand{\rels}{{\it rels}}
\newcommand{\deduce}{\vdash_p}

\newcommand{\pform}{{\rm Pr}}
\newcommand{\axform}{{\rm AX}}
\newcommand{\axset}{{\bf AX}}
\newcommand{\resdeduce}{\vdash_{\rm R}}
\newcommand{\resaxdeduce}{\vdash_{\rm R,A}}

\newcommand{\cmis}{{\em \#mis}}
\newcommand{\combine}{{\em comb}}

\newcommand{\xcsp}{{\sc X-Csp}}
\newcommand{\csp}{{\sc Csp}}

\newcommand{\cc}[1]{\textnormal{\textbf{#1}}} 
\newcommand{\opt}[0]{\textrm{{\sc opt}}}

\newcommand{\MC}{mc}
\newcommand{\Ha}{\textrm{\textit{H\aa}}}
\renewcommand{\atop}[2]{\genfrac{}{}{0pt}{}{#1}{#2}}
\newcommand{\GHAT}{{\cal G_\equiv}}
\newcommand{\HOMEQ}{\equiv}

\author{Tommy F\"arnqvist \and Peter Jonsson \and Johan Thapper}
\institute{Department of Computer and Information Science\\
Link\"opings universitet\\
SE-581 83 Link\"oping, Sweden\\
\email{\{tomfa, petej, johth\}@ida.liu.se }}

\title{Approximability Distance in the Space of $H$-Colourability Problems}
\date{}
\maketitle
\bibliographystyle{abbrv}


\begin{abstract}
A graph homomorphism is a vertex map which carries edges from a source 
graph to edges in a
target graph. We study the approximability properties of the \emph{Weighted Maximum 
$H$-Colourable Subgraph} problem ($\MCol{H}$). The instances of this problem are
edge-weighted graphs $G$ and
the objective is to find a subgraph of $G$ that has maximal total edge weight, under 
the condition that the subgraph has a homomorphism to $H$; note that for $H=K_k$ this 
problem is equivalent to {\sc Max $k$-cut}. To this end, we 
introduce a metric structure on the space of graphs
which allows us to extend previously known approximability results
to larger classes of graphs. Specifically, the approximation algorithms 
for {\sc Max cut} by Goemans and Williamson and {\sc Max $k$-cut} by Frieze and Jerrum 
can be used to yield non-trivial approximation results for $\MCol{H}$. For a variety of 
graphs, we show near-optimality results under the Unique Games Conjecture. 
We also use our method for comparing the performance of Frieze \& Jerrum's algorithm
with H\aa stad's approximation algorithm for general {\sc Max 2-Csp}.
This comparison is, in most cases, favourable to Frieze \& Jerrum.
\end{abstract}

\medskip


\section{Introduction}

\noindent
Let $G$ be a simple, undirected and finite graph. Given a subset $S \subseteq V(G)$, a 
{\em cut} in $G$ with respect to $S$ is the
edges from a vertex in $S$ to a vertex in $V(G) \setminus S$.
The {\sc Max cut}-problem asks for the size of
a largest cut in $G$.
More generally, 
a $k$-cut in $G$ is the edges going from $S_i$ to $S_j$, $i \neq j$, where
$S_1,\ldots,S_k$ is a partitioning of $V(G)$, and
the {\sc Max $k$-cut}-problem asks for the size of a largest $k$-cut.
The problem is readily seen to be identical to finding a largest $k$-colourable 
subgraph of $G$. Furthermore, {\sc Max $k$-cut} is known to be \cc{APX}-complete for 
every $k\geq 2$ and consequently does not admit a 
\emph{polynomial-time approximation scheme} ({\sc Ptas}).

In the absence of a {\sc Ptas}, it is interesting to determine the best possible
approximation ratio $c$ within which a problem can be approximated or, 
alternatively the smallest $c$ for which it can be proved that
no polynomial-time approximation algorithm exists (typically under some
complexity-theoretic assumption such as $\P \neq \NP$).
An approximation ratio of $.878567$ for {\sc Max cut} was obtained in 1995 by
Goemans and Williamson~\cite{goemans:williamson:95} 
using semidefinite
programming.
Frieze and Jerrum~\cite{frieze:jerrum:97} determined lower bounds on
the approximation ratios for {\sc Max $k$-cut} using similar techniques.
Sharpened results for small values of $k$ have later been obtained by
de Klerk et al.~\cite{deklerk:etal:04}.
Under the assumption that the {\em Unique Games Conjecture} holds, 
Khot et al.~\cite{khot:etal:2007} showed the approximation ratio for $k = 2$
to be essentially optimal and also provided upper bounds
on the approximation ratio for $k > 2$.
H\aa stad~\cite{hastad:2005} has shown that semidefinite programming is
a universal tool for solving the general
{\sc Max 2-Csp} problem over any domain, in the sense that it
establishes non-trivial approximation results for all of those problems.

In this paper, we study approximability properties of a generalised
version of {\sc Max $k$-cut} called $\MCol{H}$ for undirected graphs $H$.
Jonsson et al.~\cite{jonsson:etal:2007} have shown that, when $H$ is loop-free, $\MCol{H}$ does not admit a {\sc Ptas}. Note that if $H$ contains a loop, then $\MCol{H}$ is
a trivial problem. We present approximability results
for $\MCol{H}$ where $H$ is taken from different families of graphs. Many
of these results turns out to be close to optimal under the Unique Games Conjecture.
Our approach is based on analysing 
approximability algorithms applied to problems which they
are not originally intended to solve.
This vague idea will be clarified below.

Denote by ${\cal G}$ the set of all simple, undirected and
finite graphs. 
A \emph{graph homomorphism} $h$ from $G$ to $H$ is a vertex map which
carries the edges in $G$ to edges in $H$.
The existence of such a map will be denoted by $G \rightarrow H$.
If both $G \rightarrow H$ and $H \rightarrow G$, the graphs $G$ and $H$
are said to be \emph{homomorphically equivalent}.
This equivalence will be denoted by $G \HOMEQ H$.
For a graph $G \in {\cal G}$, let
${\cal W}(G)$ be the set of \emph{weight functions}
$w : E(G) \rightarrow {\mathbb Q}^+$ assigning weights
to edges of $G$.
For a $w \in {\cal W}(G)$, we let
$\|w\| = \sum_{e \in E(G)} w(e)$ denote the total weight of $G$.
Now,
  {\em Weighted Maximum $H$-Colourable Subgraph} (\MCol{H}) is the maximisation problem with
  \begin{description}
  \item[Instance:] An edge-weighted graph $(G,w)$, where $G \in {\cal G}$ and
    $w \in {\cal W}(G)$.

  \item[Solution:] A subgraph $G'$ of $G$ such that $G' \rightarrow H$.

  \item[Measure:] The weight of $G'$ with respect to $w$.
  \end{description}

\noindent
Given an edge-weighted graph $(G,w)$, denote by $\MC_H(G,w)$ the measure of the optimal solution to the problem \MCol{H}.
Denote by $mc_k(G,w)$ the
(weighted) size of the largest $k$-cut in $(G,w)$.
This notation is justified by the fact that
$\MC_k(G,w) = \MC_{K_k}(G,w)$.
In this sense, \MCol{H} generalises {\sc Max $k$-cut}.
The decision version of  $\MCol{H}$, the $H$-\emph{colouring} problem has been 
extensively studied (See ~\cite{HN04} and its many references.) and Hell and 
Ne\v{s}et\v{r}il~\cite{hell:nesetril:90} have shown that the problem is in $\P$ 
if $H$ contains a loop or is bipartite, and $\NP$-complete otherwise. Langberg 
et al.~\cite{langberg:etal:2006} have studied the approximability of \MCol{H} 
when $H$ is part of the input. We also note that \MCol{H} is a specialisation of the
{\sc Max Csp} problem.

The homomorphism relation $\rightarrow$ defines a quasi-order, but not a partial order on
the set ${\cal G}$.
The failing axiom is that of antisymmetry, since
$G \HOMEQ H$ does not necessarily imply $G = H$.
To remedy this, let $\GHAT$ denote the set of equivalence classes of 
${\cal G}$ under homomorphic equivalence.
The relation $\rightarrow$ is defined on $\GHAT$ in the obvious way
and on this set it is a partial order.
In fact, $\rightarrow$ provides a lattice structure on $\GHAT$ and this
lattice will be denoted by ${\cal C}_S$.
For a more in-depth treatment of graph homomorphisms and the
lattice ${\cal C}_S$, see \cite{HN04}.
Here, we endow $\GHAT$ with a metric $d$ defined in the following way:
for $M, N \in {\cal G}$, let
\begin{equation} \label{def:metric1}
d(M,N) = 1 - 
\inf_{\atop{G \in {\cal G}}{w \in {\cal W}(G)}} \frac{\MC_M(G,w)}{\MC_N(G,w)} 
\cdot
\inf_{\atop{G \in {\cal G}}{w \in {\cal W}(G)}} \frac{\MC_N(G,w)}{\MC_M(G,w)}.
\end{equation}

\noindent
In Lemma~\ref{approxresult} we will show that $d$ satisfies the following
property:
\begin{itemize}
  \item Let $M,N \in {\cal G}$ and assume that $\MC_M$ can be approximated within $\alpha$. Then, $\MC_N$ can be approximated within $(1-d(M,N)) \cdot \alpha$. 
\end{itemize}

\noindent
Hence, we can use $d$ for extending previously known approximability bounds
on $\MCol{H}$ to new and larger classes of graphs. For instance, we can
apply Goemans and Williamson's algorithm (which is intended for solving
$\MCol{K_2}$) to $\MCol{C_{11}}$ (i.e. the cycle on 11 vertices) 
and analyse how well the problem is approximated
(we will see later on that Goemans and Williamson's algorithm approximates $\MCol{C_{11}}$ within
0.79869).

In certain cases, the metric $d$ is related to a well-studied graph parameter
known as {\em bipartite density} $b(H)$~\cite{alon:96,Berman:Zhang:dm2003,bondy:locke:86,hophkins:staton:82,locke:90}: if
$H'$ is bipartite subgraph of $H$ with maximum number of edges, then 
$$b(H)=\frac{e(H')}{e(H)}.$$ 
In the end of Section~\ref{partI} we will see that $b(H) = 1 - d(K_2,H)$ for all edge-transitive graphs $H$. 
We note that while $d$ is invariant under homomorphic equivalence, this is not 
in general true for bipartite density.

The paper is divided into two main parts.
Section~\ref{partI} is used for proving the basic properties of $d$,
showing that it is well-defined on $\GHAT$, and that it is a metric.
After that, we show that $d$ is computable by linear programming
and that the computation of $d(M,N)$ can be simplified whenever
$M$ or $N$ is edge-transitive. We conclude this part by providing some examples.

The second part of the paper uses $d$ for studying the approximability of 
$\MCol{H}$. For several classes of graphs, we investigate optimiality issues by exploiting 
inapproximability bounds that are consequences of the Unique Games Conjecture. Comparisons are also made to the bounds achieved by the general {\sc Max 2-Csp}-algorithm by H\aa stad~\cite{hastad:2005}.
Our investigation covers a spectrum of graphs, ranging from graphs with few edges and/or containing large smallest cycles to graphs containing $\Theta(n^2)$ edges. Dense graphs are considered from two perspectives; firstly as graphs having a number of edges close to maximal and secondly as graphs from the ${\cal G}(n,p)$ model of random graphs, pioneered by Erd\H{o}s and R\'enyi~\cite{erdos:renyi:60}.

The techniques used in this paper seem to generalise naturally to
larger sets of problems.
This and other questions are discussed in Section~\ref{openproblems}
which concludes our paper.


\section{Approximation via the Metric $d$} \label{partI}

In this section we start out by proving basic properties of the metric $d$, that $(\GHAT, d)$ is a metric space, and that proximity of graphs $M,N$ in this space lets us interrelate the approximability of $\MCol{M}$ and $\MCol{N}$. Sections~\ref{sec:s} and~\ref{sec:lp} are devoted to showing how to compute $d$.

\subsection{The Space $(\GHAT, d)$} \label{sec:space}

We begin by introducing a function $s:{\cal G} \times {\cal G} \rightarrow \mathbb{R}$
which enables us to express $d$ in a natural way and
simplify forthcoming proofs. Let $M,N \in {\cal G}$
and define
  \begin{equation} \label{def:s}
  s(M,N) = \inf_{\atop{G \in {\cal G}}{w \in {\cal W(G)}}} \frac{mc_M(G,w)}{mc_N(G,w)}.
  \end{equation}

  \noindent
  The definition of $d$ from $(\ref{def:metric1})$ can then be written
  as follows:
  \begin{equation} \label{def:metric}
    d(M,N) = 1 - s(N,M) \cdot s(M,N).
  \end{equation}

  \noindent
  A consequence of $(\ref{def:s})$ is that the relation
  $\MC_M(G,w) \geq s(M,N) \cdot \MC_{N}(G,w)$ holds
  for all $G \in {\cal G}$ and $w \in {\cal W}(G)$.
  Using this observation, we show that $s(M,N)$ and thereby $d(M,N)$ 
  behaves well under graph homomorphisms and homomorphic equivalence.

  \begin{lemma} \label{lem:ineq2}
    Let $M, N \in {\cal G}$ and $M \rightarrow N$.
    Then, for every $G \in {\cal G}$ and every weight function
    $w \in {\cal W}(G)$,
    \[
    \MC_M(G,w) \leq \MC_N(G,w).
    \] 
  \end{lemma}

  \begin{proof}
    If $G' \rightarrow M$ for some subgraph $G'$ of $G$, then
    $G' \rightarrow N$ as well.
    The lemma immediately follows.
    \qed
  \end{proof}

  \begin{corollary} \label{cor:homom}
    If $M$ and $N$ are homomorphically equivalent, 
    then $mc_M(G,w) = mc_N(G,w)$.
  \end{corollary}

  \begin{corollary} \label{cor:shomom}
    Let $M_1 \HOMEQ M_2$ and $N_1 \HOMEQ N_2$ be two pairs of
    homomorphically equivalent graphs.
    Then, for $i,j,k,l \in \{1,2\}$, 
    \[
    s(N_i,M_j) = s(N_k,M_l).
    \]
  \end{corollary}

  \begin{proof}
    Corollary~\ref{cor:homom} shows that for all $G \in {\cal G}$ and
    $w \in {\cal W}(G)$, we have
    \[
    \frac{mc_{M_j}(G,w)}{mc_{N_i}(G,w)} =
    \frac{mc_{M_l}(G,w)}{mc_{N_k}(G,w)}.
    \]
    Now, take the infimum over graphs $G$ and weight functions $w$ on both
    sides.
    \qed
  \end{proof}


\noindent
  Corollary~\ref{cor:shomom} shows that $s$ and $d$ are well-defined as
  functions on the set $\GHAT$.
  We can now show that $d$ is indeed a metric on this space.

  \begin{lemma} \label{metric}
    The pair $(\GHAT, d)$ forms a metric space.
  \end{lemma}

  \begin{proof}
    Positivity and symmetry follows immediately from the definition
    and the fact that $s(M,N) \leq 1$ for all $M$ and $N$.
    Since $s(M,N) = 1$ if and only if $N \rightarrow M$, it also holds
    that $d(M,N) = 0$ if and only if $M$ and $N$ are homomorphically
    equivalent. That is, $d(M,N) = 0$ if and only if $M$ and $N$
    represent the same member of $\GHAT$.
    Furthermore, for graphs $M, N$ and $K \in {\cal G}$:
    \begin{multline*}
      s(M,N) \cdot s(N,K) = 
      \inf_{\atop{G \in {\cal G}}{w \in {\cal W(G)}}} \frac{mc_M(G,w)}{mc_N(G,w)} \cdot
      \inf_{\atop{G \in {\cal G}}{w \in {\cal W(G)}}} \frac{mc_N(G,w)}{mc_K(G,w)} \\ 
      \leq \inf_{\atop{G \in {\cal G}}{w \in {\cal W(G)}}} \frac{mc_M(G,w)}{mc_N(G,w)} \cdot \frac{mc_N(G,w)}{mc_K(G,w)} = s(M,K).
    \end{multline*}
    Therefore, with $a = s(M,N) \cdot s(N,M), b = s(N,K) \cdot s(K,N)$
    and $c = s(M,K) \cdot s(K,M) \geq a \cdot b$,
    \begin{multline*}
      d(M,N) + d(N,K) - d(M,K)
      = 1-a + 1-b - (1-c) \geq
      \\ \geq 1 - a - b + a \cdot b
      = (1 - a) \cdot (1 - b)
      \geq 0,
    \end{multline*}
    which shows that $d$ satisfies the triangle inequality.
    \qed
  \end{proof}

We say that a maximisation problem $\Pi$ can be approximated within $c<1$ if there
exists a randomised polynomial-time algorithm $A$ such that
$c \cdot Opt(x) \leq {\bf E}(A(x)) \leq Opt(x)$ for all instances $x$ of $\Pi$.
  Our next result shows that proximity of graphs $G$ and $H$ in $d$
  allows us to determine bounds on the approximability of \MCol{H}
  from known bounds on the approximability of \MCol{G}.

  \begin{lemma} \label{approxresult}
    Let $M,N,K$ be graphs.
    If $\MC_M$ can be approximated within $\alpha$,
    then $\MC_N$
    can be approximated within 
    $\alpha \cdot \left(1-d(M,N)\right)$.
    If it is \NP-hard to approximate $\MC_K$ within $\beta$,
    then
    $\MC_N$ is not approximable within $\beta/\left(1-d(N,K)\right)$ unless \P = \NP.
  \end{lemma}
  \begin{proof}
    Let $A(G,w)$ be the measure of the solution returned by an algorithm 
    which approximates $\MC_M$ within $\alpha$.
    We know that for all $G \in {\cal G}$ and $w \in {\cal W}(G)$
    we have the inequalities $\MC_N(G,w) \geq s(N,M) \cdot \MC_M(G,w)$
    and $\MC_M(G,w) \geq s(M,N) \cdot \MC_N(G,w)$.
    Consequently,
    \begin{multline*}
      \MC_N(G,w) \geq \MC_M(G,w) \cdot s(N,M) \geq  A(G,w) \cdot s(N,M) \\
      \geq \MC_M(G,w) \cdot \alpha \cdot s(N,M) \geq \MC_N(G,w) \cdot \alpha \cdot s(N,M) \cdot s(M,N) \\ 
      = \MC_N(G,w) \cdot \alpha \cdot (1-d(M,N)).
    \end{multline*}

    \noindent
    For the second part, assume to the contrary
    that there exists a polynomial-time algorithm $B$ that 
    approximates $\MC_N$ within $\beta/(1-d(N,K))$.
    According to the first part $\MC_K$ can then be approximated
    within $(1-d(N,K)) \cdot \beta/(1-d(N,K)) = \beta$.
    This is a contradiction unless $\P = \NP$. \qed
    
    
  \end{proof}

\subsection{Exploiting Symmetries} \label{sec:s}

We have seen that the metric $d(M,N)$ can be defined in terms of $s(M,N)$. 
  In fact, when $M \rightarrow N$ we have $1-d(M,N) = s(M,N)$.
  It is therefore of interest to find an expression for $s$
  which can be calculated easily.
  After Lemma~\ref{lem:1} (which shows how $\MC_M(G,w)$
  depends on $w$) we introduce a different way
  of describing the solutions to \MCol{M} which makes the proofs of
  the following results more natural.
  In Lemma~\ref{lem:ineq1}, we show that a particular type of weight function
  provides a lower bound on $\MC_M(G,w)/\MC_N(G,w)$.
  Finally, in Lemma~\ref{lem:orbits}, we provide a simpler expression
  for $s(M,N)$ which depends directly on the automorphism group
  and thereby the symmetries of $N$.
  This expression becomes particularly simple when $N$ is edge-transitive.  
  An immediate consequence of this is that $s(K_2,H) = b(H)$
  for edge-transitive graphs $H$. 

  The optimum $\MC_H(G,w)$ is sub-linear with respect to the
  weight function, as is shown by the following lemma. 
  \begin{lemma} \label{lem:1} 
    Let $G, H \in {\cal G}$, $\alpha \in {\mathbb Q}^+$ and let $w, w_1, \ldots, w_r \in {\cal W}(G)$ be weight functions on $G$.
    Then,
    \begin{itemize}
    \item $\MC_H(G,\alpha \cdot w) = \alpha \cdot \MC_H(G,w)$,
    \item $\MC_H(G, \sum_{i=1}^{r} w_i) \leq \sum_{i = 1}^{r} \MC_H(G,w_i)$.
    \end{itemize}
  \end{lemma}
  
  \begin{proof}
    The first part is trivial. For the second part,
    let $G'$ be an optimal solution to the instance $(G,\sum_{i=1}^r w_i)$
    of \MCol{H}.
    Then, the measure of this solution equals the sum of the measures
    of $G'$ as a (possibly suboptimal) solution to each of the
    instances $(G,w_i)$.
    \qed
  \end{proof}

  \noindent
  An alternative description of the solutions to \MCol{H} is
  as follows:
  let $G$ and $H \in {\cal G}$, and
  for any vertex map $f : V(G) \rightarrow V(H)$,
  let $f^\# : E(G) \rightarrow E(H)$ be the (partial) edge map
  induced by $f$. 
  In this notation $h : V(G) \rightarrow V(H)$ is a graph homomorphism
  precisely when $(h^\#)^{-1}(E(H)) = E(G)$ or,
  alternatively when $h^\#$ is a total function.
  The set of solutions to an instance $(G,w)$ of \MCol{H}
  can then be taken to be the set of vertex maps $f : V(G) \rightarrow V(H)$
  with the measure
  \[
  w(f) = \sum_{e \in (f^\#)^{-1}(E(H))} w(e).
  \]
  In the remaining part of this section, we will use this description
  of a solution.
  Let $f : V(G) \rightarrow V(H)$ be an optimal solution to the instance
  $(G,w)$ of \MCol{H}. 
  Define the weight $w_f \in {\cal W}(H)$ as follows:
  for each $e \in E(H)$, let
  \[
    w_f(e) = 
    \sum_{e' \in (f^\#)^{-1}(e)} \frac{w(e')}{mc_H(G,w)}.
  \]
  We now prove the following result:

  \begin{lemma} \label{lem:ineq1}
    Let $M, N \in {\cal G}$ be two graphs.
    Then, for every $G \in {\cal G}$, every $w \in {\cal W}(G)$,
    and any optimal solution $f$ to $(G,w)$ of
    \MCol{N}, it holds that
    \[
    \frac{mc_M(G,w)}{mc_N(G,w)} \geq mc_M(N,w_f).
    \]
  \end{lemma}

  \begin{proof}
    Arbitrarily choose an optimal solution $g : V(N) \rightarrow V(M)$ 
    to the instance $(N,w_f)$ of \MCol{M}.
    Then, $g \circ f$ is a solution to $(G,w)$ as an instance of 
    \MCol{M}.
    The weight of this solution is $mc_M(N,w_f) \cdot mc_N(G,w)$,
    which implies that
    \[
    mc_M(G,w) \geq mc_M(N,w_f) \cdot mc_N(G,w),
    \]
    and the result follows after division by $mc_N(G,w)$.
    \qed
  \end{proof}

  Let $M$ and $N \in {\cal G}$ be graphs and
  let $A = \Aut(N)$ be the automorphism group of $N$.
  We will let $\pi \in A$ act on $\{u,v\} \in E(N)$ by 
  $\pi \cdot \{u,v\} = \{\pi(u), \pi(v)\}$.
  The graph $N$ is edge-transitive if and only if $A$ acts
  transitively on the edges of $N$.
  Let ${\cal {\hat W}}(N)$ be the set of weight functions $w \in {\cal W}(N)$ which satisfy $\|w\| = 1$ and for which $w(e) = w(\pi \cdot e)$ for all $e \in E(N)$ and $\pi \in \Aut(N)$.

  \begin{lemma} \label{lem:orbits}
    Let $M,N\in {\cal G}$.
    Then,
    \[
    s(M,N) = \inf_{w \in {\cal {\hat W}}(N)} mc_M(N,w).
    \]
    In particular, when $N$ is edge-transitive,
    \[
    s(M,N) = mc_M(N,1/e(N)).
    \]
  \end{lemma}

  \begin{proof}
    The easy direction goes through as follows:
    \[
    s(M,N) \leq
    \inf_{w \in {\cal {\hat W}}(N)} \frac{mc_M(N,w)}{mc_N(N,w)}
  = \inf_{w \in {\cal {\hat W}}(N)} mc_M(N,w).
    \]
    \medskip

    \noindent
    For the first part of the lemma, it will be sufficient to prove that the following inequality holds for for some $w' \in {\cal {\hat W}}$.
    \begin{equation} \label{eq:ineq1b}
      \frac{\MC_M(G,w)}{\MC_N(G,w)} \geq \MC_M(N,w')
    \end{equation}
    Taking the infimum over graphs $G$ and weight functions $w \in {\cal W}(G)$ in the left-hand side of this inequality will then show that
    \[
    s(M,N) \geq \MC_M(N,w') \geq \inf_{w \in {\cal {\hat W}}(N)} mc_M(N,w).
    \]
    
    \noindent
    Let $A = \Aut(N)$ be the automorphism group of $N$.
    Let $\pi \in A$ be an arbitrary automorphism of $N$.
    If $f$ is an optimal solution to $(G,w)$ as an instance of $\MCol{N}$,
    then so is $f_\pi = \pi \circ f$.
    Let $w_\pi = w_{\pi \circ f}$.
    By Lemma~\ref{lem:ineq1}, inequality (\ref{eq:ineq1b}) is satisfied by $w_\pi$.
    Summing $\pi$ in this inequality over $A$ gives
    \[
    |A| \cdot \frac{mc_M(G,w)}{mc_N(G,w)} \geq \sum_{\pi \in A} mc_M(N,w_\pi) \geq mc_M(N,\sum_{\pi \in A} w_\pi),
    \]
    where the last inequality follows from Lemma~\ref{lem:1}.
    The weight function $\sum_{\pi \in A} w_\pi$ can be determined as follows.
    \[
      \sum_{\pi \in A} w_\pi(e) =
      \sum_{\pi \in A}
      \frac{\sum_{e' \in (f^\#)^{-1}(\pi \cdot e)} w(e')}{mc_N(G,w)} = 
      \frac{|A|}{|Ae|} \cdot
      \frac{\sum_{e' \in (f^\#)^{-1}(Ae)} w(e')}{mc_N(G,w)},
      \]
    where $Ae$ denotes the orbit of $e$ under $A$.
    Thus, $w' \sum_{\pi \in A} w_\pi / |A| \in {\cal {\hat W}}(N)$ and
    $w'$ satisfies $(\ref{eq:ineq1b})$
    so the first part follows.
    \medskip

    \noindent
    For the second part, note that when the automorphism group $A$ acts transitively on $E(N)$, there is only one orbit $Ae = E(N)$.
    Then, the weight function $w'$ is given by
    \[
    w'(e) = \frac{1}{e(N)} \cdot \frac{\sum_{e' \in (f^\#)^{-1}(E(N))} w(e')}{mc_N(G,w)} = \frac{1}{e(N)} \cdot \frac{mc_N(G,w)}{mc_N(G,w)}.
    \]
    \qed
  \end{proof}

\subsection{Tools for Computing Distances} \label{sec:lp}

  From Lemma~\ref{lem:orbits} it follows that in order to determine $s(M,N)$, 
  it is sufficient to minimise $mc_M(N,w)$ over ${\cal {\hat W}}(N)$.
  We will now use this observation to describe a linear program for computing $s(M,N)$.
  For $i \in \{ 1, \ldots, r \}$, let $A_i$ be the orbits of $\Aut(N)$
  acting on $E(N)$.
  The measure of a solution $f$ when $w \in {\cal {\hat W}}(N)$ is equal to
  $\sum_{i = 1}^r w_i \cdot f_i$, where $w_i$ is the weight of an edge in
  $A_i$ and $f_i$ is the number of edges in $A_i$ which are mapped to 
  an edge in $M$ by $f$.
  Note that given a $w$, the measure of a solution $f$ depends only on
  the vector $(f_1, \ldots, f_r) \in {\mathbb N}^r$.
  Therefore, take the solution space to be the set of such vectors:
  \[
  F = \{\, (f_1, \ldots, f_r) \;|\; \text{$f$ is a solution to $(N,w)$ of \MCol{M}} \}
  \]
  Let the variables of the linear program be 
  $w_1, \ldots, w_r$ and $s$,
  where $w_i$ represents the weight of each element in the orbit $A_i$
  and $s$ is an upper bound on the solutions.
  \[
  \begin{array}{lll}
    \min s \\
    \smallskip
    \sum_i f_i \cdot w_i \leq s & \text{for each $(f_1, \ldots, f_r) \in F$} \\
    \sum_i |A_i| \cdot w_i = 1 \\
    \smallskip
    w_i, s \geq 0 \\
  \end{array}
  \]
  
\noindent
  Given a solution $w_i, s$ to this program, $w(e) = w_i$ when $e \in A_i$ is a weight function which minimises $mc_M(G,w)$. The value of this solution is $s = s(M,N)$.

  \begin{example} \label{ex:wheel}
    The \emph{wheel graph} on $k$ vertices, $W_k$, is a graph that contains a cycle of length $k-1$ plus a vertex $v$ not in the cycle such that $v$ is connected to every other vertex. We call the edges of the $k-1$-cycle \emph{outer edges} and the remaining $k-1$ edges $\emph{spokes}$. It is easy to see that $W_k$ contains a maximum clique of size 4 if $k=4$ (in fact, $W_4=K_4$)
and a maximum clique of size 3 in all other cases. Furthermore, $W_k$ is
3-colourable if and only if $k$ is odd, and 4-colourable otherwise.
This implies that for odd $k$, the wheel graphs are homomorphically equivalent to $K_3$.

    We will determine $s(K_3, W_n)$ for even $n \geq 6$ using
    the previously described construction of a linear program.
    Note that the group action of $\Aut(W_n)$ on $E(W_n)$ has two
    orbits, one which consists of all outer edges and one which consists
    of all the spokes.
    If we remove one outer edge or one spoke from $W_k$, 
    then the resulting graph can be mapped homomorphically onto $K_3$.
    Therefore, it suffices to choose $F = \{f, g\}$ with 
    $f = (k-1,k-2)$ and $g = (k-2,k-1)$ since all other solutions
    will have a smaller measure than at least one of these.
    The program for $W_k$ looks like this:    
    \[
    \begin{array}{lll}
      \min s \\
      (k-1) \cdot w_1 + (k-2) \cdot w_2 \leq s \\
      (k-2) \cdot w_1 + (k-1) \cdot w_2 \leq s \\
      (k-1) \cdot w_1 + (k-1) \cdot w_2 = 1 \\
      w_i, s \geq 0 \\
    \end{array}
    \]
    The solution is $w_1 = w_2 = 1/(2k-2)$ with
    $s(K_3, W_k) = s = (2k-3)/(2k-2)$.
\end{example}


 \begin{example} 
    An example where the weights in the optimal solution to the
    linear program are different for different orbits is given
    by $s(K_2,K_{8/3})$.
    The {\em rational complete graph} $K_{8/3}$ has
    vertex set $\{0, 1, \ldots, 7\}$, 
    which is thought of as placed on a circle with $0$ following $7$.
    There is an edge between any two vertices which are at a distance
    at least 3 from each other.
    Each vertex has distance 4 to exactly one other vertex, which
    means there are 4 such edges.
    These edges form one orbit $A_1$ and the remaining 8 edges form
    the other orbit $A_2$.
    There are two maximal solutions, $f = (0,8)$ and $g = (4,6)$ which
    gives the program:
    \[
    \begin{array}{lll}
      \min s \\
      0 \cdot w_1 + 8 \cdot w_2 \leq s \\
      4 \cdot w_1 + 6 \cdot w_2 \leq s \\
      4 \cdot w_1 + 8 \cdot w_2 = 1 \\
      w_i, s \geq 0 \\
    \end{array}
    \]
    The solution to this program is $w_1 = 1/20$ and $w_2 = 1/10$ with
    the optimum being $4/5$.
  \end{example}

\noindent
In some cases, it may be hard to determine a desired distance between $H$ and $M$ or $N$.
If we know that $H$ is homomorphically sandwiched between
$M$ and $N$ so that $M \rightarrow H \rightarrow N$, then we can provide an upper bound on the distance of $H$ to $M$ or $N$ by using the distance between $M$ and $N$.
Formally, we have:

  \begin{lemma} \label{lemma:homdist}
    Let $M \rightarrow H \rightarrow N$.
    Then,
    \[
    s(M,H) \geq s(M,N) \qquad \textrm{and} \qquad s(H,N) \geq s(M,N).
    \]
  \end{lemma}
  
  \begin{proof}
    Since $H \rightarrow N$, it follows from Lemma~\ref{lem:ineq2} that
    $\MC_H(G,w) \leq \MC_N(G,w)$. Thus,
    \[
      s(M,H)=
      \inf_{\atop{G \in {\cal G}}{w \in {\cal W}(G)}} \frac{\MC_M(G,w)}{\MC_H(G,w)} \geq
      \inf_{\atop{G \in {\cal G}}{w \in {\cal W}(G)}} \frac{\MC_M(G,w)}{\MC_N(G,w)} = s(M,N).
      \]
    The second part follows similarly.
    \qed
  \end{proof}
We will see several applications of this lemma in Sections~\ref{subsec:sparse} and~\ref{subsec:dense}.



\section{Approximability of \MCol{\textit{H}}} \label{sec:appl}

Let $A$ be an approximation algorithm for $\MCol{H}$.
Our method basically allows us to measure how well $A$ performs on
other problems $\MCol{H'}$.
In this section, we will apply the method to various algorithms and various graphs. 
We do two things for each kind of graph under consideration: compare the performance of our method with that of some existing, leading, approximation algorithm and investigate how close to optimality we can get. Our main algorithmic tools will be the following:

\begin{theorem}[Goemans and Williamson~\cite{goemans:williamson:95}]\label{alphagw}
$mc_2$ can be approximated within \[\alpha _{GW}=\min _{0<\theta <\pi}\frac{\theta/\pi}{(1-\cos \theta )/2} \approx .878567.\]
\end{theorem}
\begin{theorem}[Frieze and Jerrum~\cite{frieze:jerrum:97}] \label{alphak}
$mc_k$ can be approximated within \[\alpha _k \sim 1-\frac{1}{k}+\frac{2\ln k}{k^2}.\]
\end{theorem}
Here, the relation $\sim$ indicates two expressions whose ratio tends to $1$ as $k \rightarrow \infty$.
We note that de Klerk et al.~\cite{deklerk:etal:04} have presented
the sharpest known bounds on
$\alpha _k$ for small values of $k$; for instance, $\alpha_3 \geq 0.836008$.

Let $v(G),e(G)$ denote the number of vertices and edges in $G$, respectively.
H\aa stad has shown the following:
\begin{theorem}[H\aa stad~\cite{hastad:2005}] \label{twocspapprox}
Let $H$ be a graph. There is an absolute constant $c>0$ such that $mc_H$ can be approximated
within 

\[1-\frac{t(H)}{d^2} \cdot (1-\frac{c}{d^2 \log d})\]

where $d=v(H)$ and $t(H)=d^2-2\cdot e(H)$.
\end{theorem}

\noindent
We will compare the performance of this algorithm on $\MCol{H}$ with the performance of the algorithms
in Theorems~\ref{alphagw} and \ref{alphak} analysed using Lemma~\ref{approxresult} and estimates
of the distance $d$.
This comparison is not entirely fair since Håstad's algorithm was probably not designed 
with the goal of providing optimal results---the goal was to beat random assignments.
However, it is the currently best algorithm that can approximate $\MCol{H}$ 
for arbitrary $H \in {\cal G}$.
For this purpose, we introduce two functions, $FJ_k$ and $\Ha$, such that, if $H$ is a graph, $FJ_k(H)$ denotes the best bound on the approximation guarantee when Frieze and Jerrum's algorithm for {\sc Max $k$-cut} is applied to the problem $\MC _H$, while $\Ha(H)$ is the guarantee when H\aa stad's algorithm is used to approximate $\MC _H$.

To be able to investigate the eventual near-optimality of our approximation method we will rely on the Unique Games Conjecture (UGC). Khot~\cite{khot:2002} suggested this conjecture as a possible direction for proving inapproximability properties of some important constraint satisfaction problems over two variables. We need the following problem only for stating the conjecture:
\begin{definition}
The Unique Label Cover problem $\mathcal{L}(V,W,E,[M],\{\pi^{v,w}\}_{(v,w)\in E})$ is the following problem: Given is a bipartite graph with left side vertices $V$, right side vertices $W$, and a set of edges $E$. The goal is to assign one `label' to every vertex of the graph, where $[M]$ is the set of allowed labels. The labelling is supposed to satisfy certain constraints given by bijective maps $\sigma _{v,w}:[M]\rightarrow [M]$. There is one such map for every edge $(v,w)\in E$. A labelling `satisfies' an edge $(v,w)$ if $\sigma _{v,w}(\mathrm{label}(w))=\mathrm{label}(v)$. The optimum of the unique label cover problem is defined to be the maximum fraction of edges satisfied by any labelling.
\end{definition}

\noindent
Now, UGC is the following:
\begin{conjecture} [Unique Games Conjecture]
For any $\eta, \gamma>0$, there exists a constant $M=M(\eta,\gamma)$ such that it is \cc{NP}-hard to distinguish whether the Unique Label Cover problem with label set of size $M$ has optimum at least $1-\eta$ or at most $\gamma$. 
\end{conjecture}

\noindent
From hereon we assume that UGC is true, which gives us the following inapproximability results:
\begin{theorem}[Khot et al.~\cite{khot:etal:2007}] \label{maxcuthard}
\begin{itemize}
\item For every $\epsilon > 0$, it is NP-hard to approximate $mc_2$ within $\alpha _{GW}+\epsilon$.
\item It is NP-hard to approximate $mc_k$ within $(1-1/k+(2 \ln k)/k^2+O((\ln \ln k)/k^2))$.
\end{itemize}
\end{theorem}

%
%

\subsection{Sparse Graphs} \label{subsec:sparse}

In this section, we investigate the performance of our method on graphs 
which have relatively few edges, and we see that the {\em girth} of the
graphs plays a central role.
The girth of a graph is the length 
of a shortest cycle contained in the graph. Similarly, the odd girth of a graph gives 
the length of a shortest odd cycle in the graph.

Before we proceed we need some facts about cycle graphs. Note that the odd cycles form a chain in the lattice ${\cal C}_S$ between $K_2$ and $C_3 = K_3$ in the following way:
\[
K_2 \rightarrow \cdots \rightarrow C_{2i+1} \rightarrow C_{2i-1} \rightarrow \cdots \rightarrow C_3 = K_3.
\]
The following lemma gives the values of $s(M,N)$ for pairs of graphs 
in this chain.
The value depends only on the target graph of the homomorphism.

\begin{lemma} \label{lem:ck}
  Let $k < m$ be positive, odd integers. Then,
  \[
  s(K_2, C_{k}) = s(C_{m}, C_{k}) = \frac{k-1}{k}.
  \]
\end{lemma}

\begin{proof}
  Note that $C_{2k+1} \not\rightarrow K_2$ and $C_{2k+1} \not\rightarrow C_{2m+1}$. However, after removing one edge from $C_{2k+1}$, the remaining subgraph is isomorphic to the path $P_{2k+1}$ which in turn is embeddable in both $K_2$ and $C_{2m+1}$.
  Since $C_{2k+1}$ is edge-transitive, the result follows from Lemma~\ref{lem:orbits}.
  \qed
\end{proof}

With Lemma~\ref{lem:ck} at hand, we can prove the following:
\begin{proposition} \label{cycleresult}
Let $k \geq 3$ be odd. Then, $FJ_2(C_k)\geq \frac{k-1}{k} \cdot \alpha_{GW}$ and $\Ha(C_k)=\frac{2}{k}+\frac{c}{k^2\log k}-\frac{2c}{k^3\log k}$. Furthermore, $\MC _{C_k}$ cannot be approximated
within $\frac{k}{k-1} \cdot \alpha_{GW}+\epsilon$ for any $\epsilon > 0$.
\end{proposition}
\begin{proof}
  From Lemma~\ref{lem:ck} we see that $s(K_2,C_k) = \frac{k-1}{k}$
  which implies (using Lemma~\ref{approxresult}) that
  $FJ_2(C_k) \geq \frac{k-1}{k} \cdot \alpha_{GW}$.
  Furthermore, $\MC_2$ cannot be approximated within $\alpha_{GW}+\epsilon'$
  for any $\epsilon' > 0$.
  From the second part of Lemma~\ref{approxresult}, we get that
  $\MC_{C_k}$ cannot be approximated within
  $\frac{k}{k-1} \cdot (\alpha_{GW} + \epsilon')$ for any $\epsilon'$.
  With $\epsilon' = \epsilon \cdot \frac{k-1}{k}$ the result follows.

Finally, we see that
  \[\Ha(C_k)
  = 1-\frac{k^2-2k}{k^2}\cdot\left(1-\frac{c}{k^2\log k}\right)
  = \frac{ck+2k^2\log k-2c}{k^3\log k} =
  \]

  \[
  = \frac{2}{k}+\frac{c}{k^2\log k}-\frac{2c}{k^3\log k}.
  \]
\end{proof}
H\aa stad's algorithm does not perform particularly well on sparse graphs; this is 
reflected by its performance on cycle graphs $C_k$ where the approximation guarantee tends to zero when $k\rightarrow \infty$. We will see that this trend is apparent for all graph types studied in this section.

Now we can continue with a result on a class of graphs with large girth:
\begin{proposition} \label{prop:lai}
  Let $m > k \geq 4$. If $H$ is a graph with odd girth $g \geq 2k+1$ and minimum
  degree $\geq \frac{2m-1}{2(k+1)}$, then $FJ_2(H)\geq\frac{2k}{2k+1} \cdot \alpha_{GW}$ and $\MC _H$ cannot be approximated within
  $\frac{2k+1}{2k} \cdot \alpha_{GW} + \epsilon$ for any $\epsilon > 0$. Asymptotically, $\Ha(H)$ is bounded by $\frac{c}{n^2\log n}+\frac{2(n^{g/(g-1)})^3}{n^4n^{1/(g-1)}}-\frac{2n^{g/(g-1)}n^{1/(g-1)}c}{n^4\log n}$, where $n=v(H)$.
\end{proposition}
\begin{proof}
Lai \& Liu~\cite{Lai:Liu:2000} have proved that if $H$ is a graph with
odd girth $\geq 2k+1$ and minimum
degree $\geq \frac{2m-1}{2(k+1)}$, then there exists a homomorphism from
$H$ to $C_{2k+1}$.
Thus, $K_2 \rightarrow H \rightarrow C_{2k+1}$ which implies that
$1-d(K_2,H) \geq 1-d(K_2,C_{2k+1}) = \frac{2k}{2k+1}$.
By Lemma~\ref{approxresult}, $FJ_2(H) \geq \frac{2k}{2k+1} \cdot \alpha_{GW}$, but $\MC _H$ cannot be approximated
within $\frac{2k+1}{2k} \cdot \alpha_{GW} +\epsilon$ for any $\epsilon > 0$.

Dutton and Brigham~\cite{dutton:brigham:91} show that one upper bound on $e(H)$ has asymptotic order $n^{1+2/(g-1)}$. This lets us say that
\[
 \Ha(H) \sim
1-\frac{n^2-2\cdot n^{1+2/(g-1)}}{n^2}\cdot\left(1-\frac{c}{n^2\log n}\right) =
\]

\[
= \frac{cn^2+2n^{(3g-1)/(g-1)}\log n-2n^{(g+1)/(g-1)}c}{n^4\log n} =
\]

\[
= \frac{c}{n^2\log n}+\frac{2(n^{g/(g-1)})^3}{n^4n^{1/(g-1)}}-\frac{2n^{g/(g-1)}n^{1/(g-1)}c}{n^4\log n}.
\]
\qed
\end{proof}

\noindent
If we restrict ourselves to planar graphs, then it is possible to show the following:
\begin{proposition} \label{prop:borodin}
Let $H$ be a planar graph with girth at least $g=\frac{20k-2}{3}$. If $v(H)=n$, then
$FJ_2(H) \geq \frac{2k}{2k+1} \cdot \alpha_{GW}$ and $\Ha(H) \leq \frac{6}{n}-\frac{12}{n^2}+\frac{c}{n^2\log n}-\frac{6c}{n^3\log n}+\frac{12c}{n^4\log n}$. $\MC_H$ cannot be approximated
within $\frac{2k+1}{2k} \cdot \alpha_{GW} +\epsilon$ for any $\epsilon > 0$.
\end{proposition}
\begin{proof}
Borodin et al.~\cite{Borodin:etal:jctb2004} have proved that $H$ is
$(2+\frac{1}{k})$-colourable which is equivalent to saying that
there exists a homomorphism from $H$ to $C_{2k+1}$.
The proof proceeds as for Proposition~\ref{prop:lai}.

The planar graph $H$ cannot have more than $3n-6$ edges so $\Ha(H)$ is bounded from above by
\[
1-\frac{n^2-2(3n-6)}{n^2}\cdot \left(1-\frac{c}{n^2\log n}\right) =
\]

\[
= \frac{cn^2-6nc+12c+6n^3\log n-12n^2\log n}{n^4\log n} =
\]

\[
= \frac{6}{n}-\frac{12}{n^2}+\frac{c}{n^2\log n}-\frac{6c}{n^3\log n}+\frac{12c}{n^4\log n}.
\]
(In fact, $H$ contains no more than $\max \{g(n-2)/(g-2),n-1\}$ edges, but using this only makes for a more convoluted expression to study.)
\qed
\end{proof}

Proposition~\ref{prop:borodin} can be strengthened and extended in different 
ways: one is to 
consider
a result by 
Dvo\v{r}\'{a}k et al.~\cite{Dvorak:etal:sidma}.
They have proved that every planar graph $H$ of odd-girth at least 9 is 
homomorphic
to the Petersen graph $P$. 
The Petersen graph is edge-transitive and
it is known (cf. \cite{Berman:Zhang:dm2003}) that the bipartite density of
$P$ is $4/5$ or, in other words, $s(K_2,P)=4/5$.
Consequently, $\MC_H$ can be approximated within
$\frac{4}{5} \cdot \alpha_{GW}$
but not within
$\frac{4}{5} \cdot \alpha_{GW} +\epsilon$ for any $\epsilon > 0$.
This is better than Proposition~\ref{prop:borodin} for planar graphs with
girth strictly less than 13.

Another way of extending Proposition~\ref{prop:borodin} is to consider graphs embeddable 
on higher-genus surfaces. For instance, the lemma is true for graphs embeddable
on the projective plane, and
it is also true for graphs of girth {\em strictly} greater than
$\frac{20k-2}{3}$ whenever the graphs are embeddable on the torus or
Klein bottle. These bounds are direct consequences of results in
Borodin et al.

%

We conclude the section by looking at a class of graphs that have small girth. 
Let $0 < \beta < 1$, be the approximation threshold for $\MC_3$, i.e.
$\MC_3$ is approximable within $\beta$ but not within $\beta + \epsilon$
for any $\epsilon > 0$. Currently, we know that $\alpha_3 \leq 0.836008 \leq \beta \leq \frac{102}{103}$
\cite{deklerk:etal:04,kann:etal:97}. 
The wheel graphs $W_k$ from Example~\ref{ex:wheel} are homomorphically equivalent to $K_3$ for odd $k$ and
we conclude (by Lemma~\ref{approxresult}) that
$\MC_{W_k}$ has the same approximability properties as $\MC_3$ in this case.
For even $k \geq 6$, $W_k$ is not homomorphically equivalent to $K_3$, though.

\begin{proposition}
  For $k \geq 6$ and even, $FJ_3(W_k) \geq \alpha_3 \cdot \frac{2k-3}{2k-2}$ but $\MC _{W_k}$ is not approximable within $\beta \cdot \frac{2k-2}{2k-3}$. $\Ha(W_k) = \frac{4}{k}-\frac{4}{k^2}+\frac{c}{k^2\log k}-\frac{4c}{k^3\log k}+\frac{4c}{k^4\log k}$.
\end{proposition}

\begin{proof}
  We know from Example~\ref{ex:wheel} that $K_3 \rightarrow W_k$ and 
  $s(K_3,W_k) = \frac{2k-3}{2k-2}$. The first part of the result follows by an application of Lemma~\ref{approxresult}.
  \[
  \Ha(W_k) =
  1-\frac{t(W_k)}{d^2}\cdot \left(1-\frac{c}{d^2\log d}\right) = 
  /d=k,e(W_k)=2(k-1)/ = 
  \]
  
  \[
  = 1-\frac{k^2-4(k-1)}{k^2}\cdot\left(1-\frac{c}{k^2\log k}\right) =
  \]
  
  \[
  = \frac{k^2c+4k^3\log k-4kc-4k^2\log k+4c}{k^4\log k} = 
  \]
  
  \[
  = \frac{4}{k}-\frac{4}{k^2}+\frac{c}{k^2\log k}-\frac{4c}{k^3\log 
    k}+\frac{4c}{k^4\log k}
  \]
  \qed
\end{proof}

\noindent
We see that $FJ_3(W_k) \rightarrow \alpha_3$ when $k \rightarrow \infty$, while $\Ha(W_k)$ tends to 0.

\subsection{Dense and Random Graphs} \label{subsec:dense}

We will now study {\em dense} graphs, i.e. graphs $H$ containing $\Theta(v(H)^2)$ edges.
For a graph $H$ on $n$ vertices, we obviously have $H\rightarrow K_n$. If we assume that 
$\omega (H)\geq r$, then we also have $K_r\rightarrow H$. Thus, if we could determine $s(K_r,K_n)$, then we could use Lemma~\ref{lemma:homdist} to calculate a bound on $FJ _n(H)$.

Let $\omega(G)$ denote the size of the largest clique in $G$ and
  $\chi(G)$ denote the chromatic number of $G$.
The Turán graph $T(n,r)$ is a graph formed by partitioning a set of $n$
vertices into $r$ subsets, with sizes as equal as possible, and connecting two
vertices whenever they belong to different subsets. Turán graphs have the following
properties \cite{turan:41}:

\begin{itemize}
\item
$e(T(n,r))=\lfloor \left(1-\frac{1}{r}\right) \cdot \frac{n^2}{2} \rfloor$; 

\item
$\omega(T(n,r))=\chi(T(n,r))=r$;

\item
if $G$ is a graph such that $e(G) > e(T(v(G),r))$, then $\omega(G) > r$.
\end{itemize}

\begin{lemma} \label{equality1}
  Let $r$ and $n$ be positive integers.
  Then,
  \[
  s(K_r,K_n) = e(T(n,r))/e(K_n) 
  \]
\end{lemma}
\begin{proof}
Since $K_n$ is edge-transitive, it suffices to show that
$\MC_r(K_n,1/e(K_n)) = e(T(n,r))/e(K_n)$.
Assume $\MC_r(K_n,1/e(K_n)) \cdot e(K_n) > e(T(n,r))$. This implies that
there exists an $r$-partite graph $G$ on $k$ vertices with strictly more than
$e(T(n,r))$ edges --- this is impossible since $\omega(G) > r$ and, consequently,
$\chi(G) > r$.
Thus, $\MC_{K_r}(K_n,1/e(K_n)) \cdot e(K_n) = e(T(n,r))$ because $T(n,r)$
is an $r$-partite subgraph of $K_n$. \qed
\end{proof}

Now, we are ready to prove the following:
\begin{proposition}
Let $v(H)=n$ and pick $r\in\mathbb{N}$, $\sigma\in\mathbb{R}$ such that \[\left \lfloor\left(1-\frac{1}{r}\right)\cdot \frac{n^2}{2} \right \rfloor  \leq \sigma\cdot n^2 = e(H)\leq \frac{n(n-1)}{2}.\]
Then,
\[ FJ_n(H) \geq \alpha _n \cdot \frac{2\left \lfloor\left(1-\frac{1}{r}\right)\cdot\frac{n^2}{2}\right \rfloor}{n\cdot(n-1)} \sim 1-\frac{1}{r}-\frac{1}{n}+\frac{2\ln n}{n(n-1)}\]
\[ \Ha (H) = 2\sigma+\frac{c}{n^2\log n}-\frac{2\sigma\cdot c}{n^2\log n}.\]
\end{proposition}
\begin{proof}
We have $K_r \rightarrow H$ due to Turán and $H \rightarrow K_n$ holds trivially since $v(H)=n$. 
By Lemma~\ref{equality1} \[ s(K_r,K_n) = \frac{2\left \lfloor\left(1-\frac{1}{r}\right)\cdot\frac{n^2}{2}\right \rfloor}{n\cdot(n-1)}. \] The first part of the result follows from Lemma~\ref{approxresult} since $d(H,K_n) \leq d(K_r,K_n)=1-s(K_r,K_n)$ and some straightforward calculations.
\[
\Ha (H)
= 1-\frac{n^2-\sigma\cdot n^2}{n^2}\cdot \left(1-\frac{c}{n^2\log n} \right) =
\]

\[
= \frac{c+2\sigma\cdot n^2\log n - 2\sigma \cdot c}{n^2\log n}
= \frac{c}{n^2\log n}+2\sigma-\frac{2\sigma \cdot c}{n^2\log n}.
\]
\qed
\end{proof}
Note that when $r$ and $n$ grow, $FJ_n(H)$ tends to $1$. This means that, asymptotically, we cannot do much better. If we compare the expression for $FJ_n(H)$ with the inapproximability bound for $\MC _n$ (Theorem~\ref{maxcuthard}), we see that all we could hope for is a
faster convergence towards $1$. As $\sigma$ satisfies $\left(1-\frac{1}{r}\right)\cdot\frac{1}{2} \leq \sigma \leq \left(1-\frac{1}{n}\right)\cdot\frac{1}{2}$, we conclude that $\Ha (H)$ also tends to $1$ as $r$ and $n$ grow. To get a better grip on how $\Ha (H)$ behaves we look at two extreme cases.

For a maximal $\sigma=\left(1-\frac{1}{r}\right)\cdot\frac{1}{2}$, $\Ha (H)$ becomes \[1-\frac{1}{n}+\frac{c}{n^3\log n}. \] On the other hand, this guarantee, for a minimal $\sigma=\left(1-\frac{1}{r}\right)\cdot\frac{1}{2}$ is \[1-\frac{1}{r}+\frac{c}{rn^2\log n}.\] At the same time, it is easy to see that Frieze and Jerrum's algorithm makes these points approximable within $\alpha _n$ (since, in this case, $H \equiv K_n$) and $\alpha _r$ (since Tur\'an's theorem tells us that $H\rightarrow K_r$ holds in this case), respectively. Our conclusion is that Frieze and Jerrum's and H\aa stad's algorithms perform almost equally well on these graphs asymptotically.
 
Another way to study dense graphs is via random graphs. Let ${\cal G}(n,p)$ denote the random graph on $n$ vertices in which every edge is chosen randomly and independently with probability $p=p(n)$. We say that ${\cal G}(n,p)$ has a property $A$ \emph{asymptotically almost surely} (a.a.s.) if the probability it satisfies $A$ tends to $1$ as $n$ tends to infinity. Here, we let $p=c$ for some $0<c<1.$

For $G\in {\cal G}(n,p)$ it is well known that a.a.s. $\omega (G)$ assumes one of at most two values around $\frac{2\ln n}{\ln (1/p)}$~\cite{bollobas:erdos:76,matula:72}. It is also known that, almost surely $\chi (G)\sim \frac{n}{2\ln (np)}\ln \left(\frac{1}{1-p}\right)$, as $np\rightarrow\infty$~\cite{bollobas:88,luczak:91}. Let us say that $\chi (G)$ is concentrated in width $s$ if there exists $u=u(n,p)$ such that a.a.s. $u\leq\chi (G)\leq u+s$. Alon and Krivelevich~\cite{alon:krivelevich:97} have shown that for every constant $\delta > 0$, if $p=n^{-1/2-\delta}$ then $\chi (G)$ is concentrated in width $s=1$. That is, almost surely, the chromatic number takes one of two values.

\begin{proposition}
Let $H\in {\cal G}(n,p)$. When $np \rightarrow \infty$, $FJ_m(H) \sim 1-\frac{2}{m}+\frac{2\ln m}{m^2}+\frac{1}{m^2}-\frac{2\ln m}{m^3}$, where $m=\omega (H)$. $\Ha (H) = p-\frac{p}{n}+(1-p)\cdot \frac{c}{n^2\log n}+\frac{pc}{n^3\log n}$.
\end{proposition}
\begin{proof}
Let $k=\chi (H)$.

\[
FJ_m(H) \geq \alpha _m\cdot s(K_m,K_k) \sim \left(1-\frac{1}{m}+\frac{2\ln m}{m^2}\right)\cdot\frac{2\left \lfloor\left(1-\frac{1}{m}\right)\cdot\frac{k^2}{2}\right \rfloor}{k(k-1)}
\sim 
\]

\[
\sim \frac{(m^2-m+2\ln m)(m-1)k}{m^3(k-1)} =
\]

\[
= \frac{k}{k-1}-\frac{2k}{m(k-1)}+\frac{k}{m^2(k-1)}+\frac{2k\ln m}{m^2(k-1)}-\frac{2k\ln m}{m^3(k-1)} \ (**)
\]

\noindent
If $n$ is large, then $k \gg m$ and

\[
(**) \sim 1-\frac{2}{m}+\frac{2\ln m}{m^2}+\frac{1}{m^2}-\frac{2\ln m}{m^3}.
\] 

\noindent
The expected number of edges for a graph $H\in {\cal G}(n,p)$ is $\binom{n}{2}p$, so
\[
\Ha (H)
= 1-\frac{t(G)}{d^2}\cdot (1-\frac{c}{d^2\log d})=/d=n, e(G)=\binom{n}{2}p/ =
\]

\[
= 1-\frac{n^2-(n^2-n)p}{n^2}\cdot (1-\frac{c}{n^2\log n})
=  1-\frac{n-pn+p}{n}\cdot (1-\frac{c}{n^2\log n}) =
\]

\[
= 1-(1-p+\frac{p}{n})\cdot (1-\frac{c}{n^2\log n}) =
\]

\[
=  \frac{pn^3\log n + nc - pnc - pn^2\log n + pc}{n^3 \log n} =
\]

\[
= p-\frac{p}{n}+(1-p)\cdot \frac{c}{n^2\log n}+\frac{pc}{n^3\log n}
\]
\qed
\end{proof}
We see that, in the limiting case, $\Ha (H)$ tends to $p$, while $FJ_m(H)$ tends to $1$. Again, this means that, for large enough graphs, we cannot do much better. With a better analysis, one could possibly reach an expression for $FJ_m(H)$ that has a faster convergence rate.

Of course, it is interesting to look at what happens for graphs $H\in {\cal G}(n,p)$ where $np$ does not tend to $\infty$ when $n\rightarrow\infty$. The following theorem lets us do this.
\begin{theorem}[Erd\H{o}s and R\'enyi~\cite{erdos:renyi:60}]
Let $c$ be a positive constant and $p=\frac{c}{n}$. If $c<1$, then a.a.s. no component in ${\cal G}(n,p)$ contains more than one cycle, and no component has more than $\frac{\ln n}{c-1-\ln c}$ vertices.
\end{theorem}
Now we see that if $np\rightarrow \epsilon$ when $n\rightarrow\infty$ and $0<\epsilon<1$, then ${\cal G}(n,p)$ almost surely consists of components with at most one cycle. Thus, each component resembles a cycle where, possibly, trees are attached to certain cycle vertices, and each component is homomorphically equivalent to the cycle it contains. Since we know from Section~\ref{subsec:sparse} that Frieze and Jerrum's algorithm performs better than H\aa stads algorithm on cycle graphs, it follows that the same relationship 
holds in this part of the ${\cal G}(n,p)$ spectrum.

\section{Conclusions and Open Problems} \label{openproblems}

We have seen that applying Frieze and Jerrum's algorithm to $\MCol{H}$ gives 
comparable to or better results than when applying H\aa stad's {\sc Max 2-Csp} 
algorithm for the classes of graphs we have considered. One possible explanation 
for this is that the analysis of the {\sc Max 2-Csp} algorithm only aims to 
prove it better than a random solution on expectation, which may leave room for 
strengthenings of the approximation guarantee. At the same time, we are 
probably overestimating the distance between the graphs. It is likely that
both results can be improved.

Kaporis et al.~\cite{kaporis:etal:2006} have shown that $\MC _2$ is 
approximable within $.952$ for any given average degree $d$ and asymptotically 
almost all random graphs $G$ in ${\cal 
G}(n,m=\left\lfloor\frac{d}{2}n\right\rfloor)$, where ${\cal G}(n,m)$ is the 
probability space of random graphs on $n$ vertices and $m$ edges selected 
uniformly at random. In a similar vein, Coja-Oghlan et al.~\cite{coja:etal:2005} 
give an algorithm that approximates $\MC _k$ within $1-O(1/\sqrt{np})$ in 
expected polynomial time, for graphs from ${\cal G}(n,p)$. It would be 
interesting to know if these results could be carried further, to other graphs 
$G$, so that better approximability bounds on $\MCol{H}$, for $H$ such that 
$G\rightarrow H$, could be achieved.

Erd\H{o}s \cite{erdos:59} has proved that for any positive
integers $k$ and $l$
there exists a graph of chromatic number $k$ and girth at least $l$.
It is obvious that such graphs cannot be sandwiched between
$K_2$ and a cycle as was the case of the graphs of high girth in
Section~\ref{subsec:sparse}.
A different idea is thus required to deal with these graphs.
In general, to apply our method more precisely, we need a better
understanding of the structure of ${\cal C}_S$ and how this
interacts with our metric $d$.

The idea of defining a metric on a space of problems which relates
their approximability can be extended to more general cases.
It should not prove too difficult to generalise the framework introduced
in this paper to {\sc Max CSP} over directed graphs or even
languages consisting of a single, finitary relation.
How far can this generalisation be carried out?
Could it provide any insight into the approximability of 
{\sc Max CSP} on arbitrary constraint languages?

When considering inapproximability, we have strongly relied on 
the Unique Games Conjecture---hence, we are part of the growing body interested
in seeing UGC settled. We note, though, that
weaker inapproximability results exist for both {\sc Max 
cut}~\cite{hastad:2001} and {\sc Max $k$-cut}~\cite{kann:etal:97} and that they are
applicable in our setting. We want to emphasise that our method is not 
{\em per se} dependent on the truth of the UGC.

\bibliography{maxhcol_long,/home/chrba/bib/strings,/home/chrba/bib/abstracts,/home/chrba/bib/papers,/home/chrba/bib/books,/home/chrba/bib/xref,/home/petej/bib/extrapapers}
\end{document}